\documentclass[letterpaper, 10 pt, conference]{ieeeconf}
\IEEEoverridecommandlockouts

\usepackage{amsmath, amsthm, amssymb, amsfonts, mathtools, xargs, tensor, units, cite, stmaryrd, mathrsfs, algpseudocode, comment, graphicx}
\usepackage[unicode=true, bookmarks=true, bookmarksnumbered=false, bookmarksopen=false, breaklinks=false, pdfborder={0 0 1}, backref=false, colorlinks=false, hidelinks]{hyperref}
\makeatletter
\let\NAT@parse\undefined
\makeatother
\newcommand{\dif}{\mathop{}\!\mathrm{d}}
\usepackage[colorinlistoftodos]{todonotes}
\theoremstyle{defn}
\newtheorem{defn}{Definition}

\theoremstyle{assumption}
\newtheorem{assumption}{Assumption}

\theoremstyle{lem}

\theoremstyle{thm}
\newtheorem{thm}{Theorem}

\theoremstyle{remark}
\newtheorem{remark}{Remark}

\begin{document}
	\title{Online inverse reinforcement learning with unknown disturbances}
	
	\author{Ryan Self, Moad Abudia and Rushikesh Kamalapurkar\thanks{The authors are with the School of Mechanical and Aerospace Engineering, Oklahoma State University, Stillwater, OK, USA. {\tt\small \{rself, abudia, rushikesh.kamalapurkar\}@okstate.edu}.}}
	\maketitle
	\begin{abstract} 
		This paper addresses the problem of online inverse reinforcement learning for nonlinear systems with modeling uncertainties while in the presence of unknown disturbances. The developed approach observes state and input trajectories for an agent and identifies the unknown reward function online. Sub-optimality introduced in the observed trajectories by the unknown external disturbance is compensated for using a novel model-based inverse reinforcement learning approach. The observer estimates the external disturbances and uses the resulting estimates to learn the dynamic model of the demonstrator. The learned demonstrator model along with the observed suboptimal trajectories are used to implement inverse reinforcement learning. Theoretical guarantees are provided using Lyapunov theory and a simulation example is shown to demonstrate the effectiveness of the proposed technique. 
		
	\end{abstract}
	\section{Introduction}
	
	Based on the premise that the most succinct representation of the behavior of an entity is its reward structure \cite{SCC.Ng.Russell2000}, this paper aims to recover the reward (or cost) function of an agent by observing the agent performing a task and monitoring its state and control trajectories. The reward function estimation is performed in the presence of modeling uncertainties and unknown disturbances. This process of learning an agent's reward function is known as inverse reinforcement learning (IRL) \cite{SCC.Ng.Russell2000,SCC.Russell1998}.
	
	IRL methods are proposed in \cite{SCC.Ng.Russell2000} and reward function estimation techniques using IRL for problems formulated as Markov Decision Processes (MDP) are shown in \cite{SCC.Abbeel.Ng2004,SCC.Abbeel.Ng2005,SCC.Ratliff.Bagnell.ea2006}. Since solutions to the IRL problems are generally not unique, the maximum entropy principle is developed in \cite{SCC.Ziebart.Maas.ea2008} to help differentiate between the various solutions. In \cite{SCC.Zhou.Bloem.ea2018}, the authors develop a Maximum Causal Entropy IRL technique for infinite time horizon problems where a stationary soft Bellman policy which helps enable the learning of the transition models is utilized. Beyond this, many extensions of IRL include the formulation of feature constructions \cite{SCC.Levine.Popovic.ea2010}, solving IRL using gradient based methods \cite{SCC.Neu.Szepesvari2007}, and game theoretic approaches \cite{SCC.Syed.Schapire2008}, which suggest the possibility of finding solutions that outperform the expert. IRL is also extended to nonlinear problems using Gaussian Processes, such as \cite{SCC.Levine.Popovic.ea2011}.
	
	The aforementioned IRL techniques and inverse optimal control methods \cite{SCC.Kalman1964} are extensively utilized to teach autonomous machines to perform specific  tasks in an \emph{offline} setting \cite{SCC.Mombaur.Truong.ea2010}. However, these \emph{offline} approaches do not have the robustness to uncertainties required for online implementation. Inspired by the success of model-based real-time reinforcement learning methods such as in \cite{SCC.Vamvoudakis.Lewis2010} and \cite{SCC.Wang.Liu.ea2016} and the online IRL/Inverse Optimal Control (IOC) results for linear systems in \cite{SCC.Kamalapurkar2018} and \cite{SCC.Molloy.Ford.ea2018}, this paper presents an IRL technique for nonlinear systems. In this paper, the results of \cite{SCC.Self.Harlan.ea2019a} are extended to address the problem of online IRL in the presence of disturbances by developing a recursive IRL technique.

	The main contribution of this paper is the development of
	a novel method for reward function estimation for an agent
	with unknown dynamics in the presence of disturbances. The
	developed technique in this paper builds on the previous
	work in \cite{SCC.Self.Harlan.ea2019a} where a batch IRL method is utilized that relies
	on optimal demonstrations, and as such, does not consider
	external disturbances affecting the agent being observed. The
	recursive IRL update results in smoother weight estimates
	and admits Lyapunov-based performance guarantees. Addressing
	the complexities resulting for disturbance-induced
	sub-optimality of the demonstrations is one of the major
	technical contributions of this paper. The suboptimal observations
	make model-free IRL methods challenging because
	they are entirely trajectory driven, and in general, require
	either optimal or near optimal observations. The novelty
	in the technique developed in this paper is the use of
	model-based IRL to compensate for the disturbance-induced
	sub-optimality. If dynamic models of the agents under observation
	are unavailable, they need to be learned from
	data. However, the disturbances make system identification
	challenging, and the resulting models are typically poor. To
	overcome this challenge, it is assumed that the observer and
	demonstrator are co-located and as a result, experience the
	same disturbance. One can then learn the disturbances using
	their effects on the observer and use the resulting estimates
	to learn the dynamic model of the agent under observation.
	A model-based IRL method can then be deployed to learn
	the unknown reward function.
	
	The paper is organized as follows: Section \ref{sec:Notation} explains the notation used throughout the paper. Section \ref{sec:Problem Formulation} details the problem formulation and how the additional challenges related to disturbances are addressed. Section \ref{sec:Dist_Estimation} details the disturbance estimator for this method. Section \ref{sec: Param_Est} shows the developed parameter estimator. 
	Section \ref{sec: IRL} explains the IRL algorithm. 
	Section \ref{sec:Sim} shows a simulation example for the proposed method and Section \ref{sec:Conclusion} concludes the paper. 
	
	\section{Notation}\label{sec:Notation}
	The notation $\mathbb{R}^{n}$ represents the $n-$dimensional Euclidean space, and the elements of $\mathbb{R}^{n}$ are interpreted as column vectors, where $\left(\cdot\right)^{T}$ denotes the vector transpose operator. The set of positive integers excluding 0 is denoted by $\mathbb{N}$. For $a\in\mathbb{R},$ $\mathbb{R}_{\geq a}$ denotes the interval $\left[a,\infty\right)$, and $\mathbb{R}_{>a}$ denotes the interval $\left(a,\infty\right)$. If $a\in\mathbb{R}^{m}$ and $b\in\mathbb{R}^{n}$, then $\left[a;b\right]$ denotes the concatenated vector $\begin{bmatrix}a\\
	b
	\end{bmatrix}\in\mathbb{R}^{m+n}$. The notations $\text{I}_{n}$ and $0_{n}$ denote the $n\times n$ identity matrix and the zero element of $\mathbb{R}^{n}$, respectively. Whenever it is clear from the context, the subscript $n$ is suppressed.
	
	\section{Problem Formulation} \label{sec:Problem Formulation}
	Consider two agents, Agent 1 and Agent 2, where Agent 1 is monitoring the behavior of Agent 2. Agent 1 has the following dynamics
	\begin{align}
	\dot{x}_1=f_1(x_1,u_1)+d_1, \label{eq: Unknown Dist1}
	\end{align}
	where $x_1:\mathbb{R}_{\geq T_0}\rightarrow \mathbb{R}^n$ is the state, $u_1:\mathbb{R}_{\geq T_0}\rightarrow \mathbb{R}^m$ is the control, $f_1:\mathbb{R}^n\times\mathbb{R}^m\rightarrow \mathbb{R}^n$ are the dynamics, $d_1:\mathbb{R}_{\geq T_0}\rightarrow \mathbb{R}^n$ is a disturbance acting on Agent 1, and $T_0$ is the initial time. The dynamics for Agent 2 are
	\begin{align}
	\dot{x}_2=f_2(x_2,u_2)+d_2, \label{eq: Unknown Dist2}
	\end{align}
	where $x_2:\mathbb{R}_{\geq T_0}\rightarrow \mathbb{R}^n$ is the state, $u_2:\mathbb{R}_{\geq T_0}\rightarrow \mathbb{R}^m$ is the control, $f_2:\mathbb{R}^n\times\mathbb{R}^m\rightarrow \mathbb{R}^n$ are the dynamics, and $d_2:\mathbb{R}_{\geq T_0}\rightarrow \mathbb{R}^n$ is a disturbance acting on the Agent 2. 
	
	Agent 2 is attempting to follow a policy that minimizes the following performance index
	\begin{equation}
	J(x_0,u(\cdot))= \int_{T_0}^{\infty}r(x(t;x_0,u(\cdot)),u(t))\dif t,\label{eq:reward}
	\end{equation} 
	where $x\left(\cdot;x_0,u(\cdot)\right)$ is the trajectory generated by the optimal controller $u(\cdot)$ for the undisturbed dynamics that minimizes the performance index in \eqref{eq:reward} starting at $x_0$ and beginning at time $T_0$. The main objective of this paper is to estimate the unknown reward function, $r$, in the presence of disturbances and uncertainties in the dynamics.
	
	The following assumptions are used in the analysis of the paper.
	\begin{assumption}
		The disturbances affecting both agents are identical, i.e. $d_1\left(t\right) = d_2\left(t\right)=d(t), \forall t$.
	\end{assumption}
	\begin{assumption}
		The unknown reward function $r$ is quadratic in the control, i.e.,
		\begin{equation}
		r(x,u) = Q(x)+u^TRu,
		\end{equation}
		where $R\in \mathbb{R}^{m\times m}$ is a positive definite matrix, such that $R=\mathrm{diag}([r_1,\cdots,r_m])$, and the function $Q$ can be represented using a neural network as $Q(x) =(W^*_Q)^T\sigma_{Q}(x)+\epsilon_Q(x)$ is a positive definite function, where $ W_{Q}^{*}\coloneqq\left[q_{1},\cdots,q_{L}\right]^{T} $ are ideal reward function weights, $\sigma_{Q}: \mathbb{R}^{n} \to \mathbb{R}^L$ are known continuously differentiable features, and $\epsilon_Q:\mathbb{R}^n\to \mathbb{R}$ is the approximation error.
	\end{assumption}
	\begin{assumption}
		The dynamics for Agent 2 can be expressed as
		\begin{equation}
		\dot{x}_2=f_2^o(x_2,u_2)+\theta_2^T\sigma_2(x_2,u_2)+{d}, \label{eq:NLSPre Integral Form}
		\end{equation}
		where $f_2^o:\mathbb{R}^n\times\mathbb{R}^m\rightarrow \mathbb{R}^n$ denotes the nominal dynamics, $\theta_2^T\sigma_2$ is a parameterized estimate of the uncertain part of the dynamics, $\theta_2 \in \mathbb{R}^{p\times n}$ is a matrix of unknown constant parameters, and $\sigma_2:\mathbb{R}^{n}\times \mathbb{R}^m\to \mathbb{R}^p$ are known features.
	\end{assumption}

	If Agent 1 and Agent 2 are co-located and of similar size then the disturbances affecting them can be reasonably assumed to be equal. Assumption 2 facilitates the IRL problem formulation in Section \ref{sec: IRL}, and Assumption 3 facilitates the parameter estimation in Section \ref{sec: Param_Est}.
	
	A solution to the two-agent IRL problem described above is proposed in the following. 
	
	Due to the unknown disturbance $d$ acting on Agent 2, the observed trajectories corresponding to Agent 2 will no longer be optimal with respect to its unknown performance index. As a result, a purely data-driven implementation of IRL would yield incorrect reward function estimates. Instead, in this paper, the state trajectories for Agent 2 are measured and the reward function is estimated using a model-based approach that compensates for the trajectory deviations. The unknown disturbance, $d$, is estimated by Agent 1 using its known internal model, and Agent 1 implements a parameter estimator that incorporates the disturbance estimates to calculate the unknown parameters in the dynamics of Agent 2. Finally, both the disturbance and parameter estimates are used by Agent 1 to estimate the unknown reward function that Agent 2 is acting with respect to.
	
	The following sections; disturbance estimation, parameter estimation, and inverse reinforcement learning, are performed in parallel and in real-time.
	
	\section{Disturbance Estimation}\label{sec:Dist_Estimation}
	While the IRL method discussed in the following can be developed using any disturbance estimator that results in uniform ultimate boundedness of the disturbance estimation error, the following exponential disturbance estimator (inspired by \cite{SCC.Chen2004}) is used in this paper for ease of exposition. The disturbance estimation is performed only by Agent 1, and for clarity, the subscripts in the dynamics will be omitted in this section.
	
	The unknown disturbance acting on the agents is assumed to be an additive disturbance that is generated from the exogenous linear system
	\begin{equation}
	\dot{\zeta}=A \zeta,\label{eq:Dist_System}
	\end{equation}
	\begin{equation}
	{d}=C \zeta, \label{eq:d Equation}
	\end{equation}
	where $\zeta:\mathbb{R}_{\geq T_0}\rightarrow \mathbb{R}^N, A \in\mathbb{R}^{N\times N},C\in\mathbb{R}^{n\times N},$ and $d:\mathbb{R}_{\geq T_0}\rightarrow \mathbb{R}^n$ is the disturbance.
	
	The disturbance estimator is designed as
	\begin{equation}
	\dot{\hat{\zeta}}=A \hat{\zeta}+K\left(
	\dot{x}- \left(f\left(x,u\right)+ \hat{d}\right)\right),\label{eq:Dist_Est}
	\end{equation}
	and
	\begin{equation}
	\hat{d}=C \hat{\zeta}, \label{eq: dhat Equation}
	\end{equation} 
	where $K\in \mathbb{R}^{N\times n}$ is a gain matrix.
	
	The following theorem utilizes Lyapunov-based arguments to establish exponential convergence of the disturbance estimator.
	
	\begin{thm}
		If $(A-KC)$ is negative definite, then the disturbance estimation error converges exponentially to zero.
	\end{thm}
	\begin{proof}
		Define the error for the exogenous linear system as
		\begin{equation}
		\tilde{\zeta} = \zeta - \hat{\zeta}.\label{eq:error}
		\end{equation}
		Consider the positive definite candidate Lyapunov function
		\begin{equation}
		V_d(\tilde{\zeta}) = \frac{1}{2}\tilde{\zeta}^T\tilde{\zeta}. \label{eq:V_disterror}
		\end{equation}
		Taking the time-derivative of \eqref{eq:V_disterror}, using \eqref{eq:Dist_System}, and \eqref{eq:Dist_Est}
		\begin{multline}
		\dot{V}_d(\tilde{\zeta}) = \tilde{\zeta}^T\Big(A\zeta-A\hat{\zeta}-K\Big(
		\dot{x}-\Big(f\left(x,u\right) + \hat{d}\Big)\Big)\Big).\label{eq:VDot_Dist}
		\end{multline}
		Using \eqref{eq: Unknown Dist2} and simplifying, results in
		\begin{equation}
		\dot{V}_d(\tilde{\zeta}) = \tilde{\zeta}^T\left(A\tilde{\zeta}-K\left(d-\hat{d}\right)\right).
		\end{equation}
		Using \eqref{eq:d Equation} and \eqref{eq: dhat Equation}
		\begin{equation}
		\dot{V}_d(\tilde{\zeta})=\tilde{\zeta}^T\left(A-KC\right)\tilde{\zeta}.\label{eq:VDot_DistFinal}
		\end{equation}
		Using \eqref{eq:VDot_DistFinal}, provided $A-KC$ is negative definite, it can be concluded that $\tilde{\zeta}$ converges exponentially to zero. Since $\tilde{d}=C\tilde{\zeta}$, $\tilde{d}$ has the same convergence rate as $\tilde{\zeta}.$
	\end{proof}
	\section{Parameter Estimation}\label{sec: Param_Est}
	A parameter estimator, motivated by the authors' previous work  in \cite{SCC.Kamalapurkar2017a}, is developed in this section. Since parameter estimation is performed only for Agent 2, for clarity, the subscripts for the dynamics will be omitted in this section. 
	
	\subsection{Parameter Estimator}\label{sec: Param Estim}
	Integrating \eqref{eq:NLSPre Integral Form} over the interval $\left[t-T,t\right]$
	for some constant $T\in\mathbb{R}_{>0}$,\footnote{If the integration interval is selected to be too short, there may not be enough information in the vector $X_i$ to achieve accurate parameter estimation. If the integration interval is selected too long, parameter estimates may not be available during transients where they are needed the most. The development of a reasonable heuristic that guides the selection of the integration interval is a topic for future research.}
	\begin{multline}
	x\left(t\right)-x\left(t-T\right)=\int_{t-T}^{t}f^{o}\left(x(\gamma),u(\gamma)\right)\dif\gamma\\+\theta^{T}\int_{t-T}^{t}\sigma\left(x(\gamma),u(\gamma)\right)\dif\gamma+\int_{t-T}^{t}{d}\left(\gamma\right)\dif\gamma.\label{eq:NLSDouble Integral Form}
	\end{multline}
	
	The expression in (\ref{eq:NLSDouble Integral Form}) can be rearranged to form the affine system
	\begin{equation}
	X\left(t\right)=F\left(t\right)+\text{\ensuremath{\theta}}^{T}S\left(t\right)+{D}\left(t\right),\:\forall t\in\mathbb{R}_{\geq T_{0}}\label{eq:NLSDerivative Free Form}
	\end{equation}
	where 
	\begin{equation}
	X\left(t\right):=\begin{cases}
	\begin{gathered}x(t)-x\left(t-T\right),
	\end{gathered}
	& t\in\left[T_{0}+T,\infty\right),\\
	0, & t<T_{0}+T,
	\end{cases}\label{eq:NLSP}
	\end{equation}
	\begin{equation}
	F\left(t\right):=\begin{cases}
	\int_{t-T}^{t}f^{o}\left(x(\gamma),u(\gamma)\right)\dif\gamma, & t\in\left[T_{0}+T,\infty\right),\\
	0, & t<T_{0}+T,
	\end{cases}\label{eq:NLSF}
	\end{equation}
	\begin{equation}
	S\left(t\right):=\begin{cases}
	\int_{t-T}^{t}\sigma\left(x(\gamma),u(\gamma)\right)\dif\gamma, & t\in\left[T_{0}+T,\infty\right),\\
	0, & t<T_{0}+T,
	\end{cases}\label{eq:NLSG}
	\end{equation}
	and
	\begin{equation}
	{D}\left(t\right):=\begin{cases}
	\int_{t-T}^{t}{d}\left(\gamma\right)\dif\gamma,
	& t\in\left[T_{0}+T,\infty\right),\\
	0, & t<T_{0}+T.
	\end{cases}\label{eq:NLSD}
	\end{equation}
	
	The affine error system in (\ref{eq:NLSDerivative Free Form}) motivates
	the adaptive estimation scheme that follows, in which a \emph{concurrent learning} \cite{SCC.Chowdhary2010} technique is developed that utilizes recorded data stored in a history stack to drive parameter estimation.
	
	A history stack, $\mathcal{H}^{PE}$, is a set of data points $\left\{ \left(X_{i},{F}_{i},{S}_{i},\hat{D}_{i}\right)\right\} _{i=1}^{M}$ such that \begin{equation}
	X_{i}={F}_{i}+\text{\ensuremath{\theta}}^{T}{S}_{i}+\hat{D}_{i}+\mathcal{E}_i,\:\forall i\in\left\{ 1,\cdots,M\right\}, \label{eq:NLSHistory Stack Compatibility}
	\end{equation}
	where $\mathcal{E}_i=D_i-\hat{D}_i$, and 
	\begin{equation}
	\hat{D}\left(t\right):=\begin{cases}
	\int_{t-T}^{t}\hat{d}\left(\gamma\right)\dif\gamma,
	& t\in\left[T_{0}+T,\infty\right),\\
	0, & t<T_{0}+T.
	\end{cases}\label{eq:NLSDHat}
	\end{equation}
	$\mathcal{H}^{PE}$ is called \emph{full rank} if there exists a constant $\underline{c}\in\mathbb{R}$ such that \begin{equation} 0<\underline{c}<\lambda_{\min}\left\{ \mathscr{S}\right\} ,\label{eq:NLRank Condition} \end{equation} where the matrix  $\mathscr{S}\in\mathbb{R}^{p\times p}$ is defined as $\mathscr{S}:=\sum_{i=1}^{M}{S}_{i}{S}_{i}^{T}$. The concurrent learning update law to estimate the unknown parameters is then given by
	\begin{equation}
	\dot{\hat{\theta}}=\alpha_{\theta}\Gamma_\theta\sum_{i=1}^{M}{S}_{i}\left(X_{i}-{F}_{i}-\hat{\theta}^{T}{S}_{i}-\hat{D}_{i}\right)^{T},\label{eq:NLSTheta Dynamics}
	\end{equation}
	where $\alpha_{\theta}\in\mathbb{R}_{>0}$ is a constant adaptation gain, and $\Gamma_\theta:\mathbb{R}_{\geq0}\to\mathbb{R}^{p\times p}$ is the least-squares gain updated using the update law
	\begin{equation}
	\dot{\Gamma}_{\theta}=\beta_{\theta}\Gamma_{\theta}-\alpha_{\theta}\Gamma_{\theta}\mathscr{S}\Gamma_{\theta},\label{eq:NLGamma Dynamics}
	\end{equation}
	where $\beta_{\theta}\in\mathbb{R}_{>0}$ is a constant gain. Using arguments similar to \cite[Corollary 4.3.2]{SCC.Ioannou.Sun1996}, it can be shown that provided $\lambda_{\min}\left\{ \Gamma_{\theta}^{-1}\left(0\right)\right\} >0$, the least squares gain matrix satisfies
	\begin{equation}
	\underline{\Gamma}_{\theta}\text{I}_{p}\leq\Gamma_{\theta}\left(t\right)\leq\overline{\Gamma}_{\theta}\text{I}_{p},\label{eq:NLSStaFGammaBound_theta}
	\end{equation}
	where $\underline{\Gamma}_{\theta}$ and $\overline{\Gamma}_{\theta}$ are positive
	constants, and $\text{I}_{p}$ denotes an $p\times p$ identity matrix. If a full rank history stack that satisfies \eqref{eq:NLSHistory Stack Compatibility} is not available a priori, then the data points can be recorded online.
	
	From the Lyapunov analysis in Section \ref{sec: Ana PE}, it is observed that the rate of decay for the parameter estimation error is proportional to the minimum singular value of $\mathscr{S}$. Therefore, to promote faster convergence for the parameter estimates, a minimum singular value maximization algorithm is developed. At each time $t$, the algorithm takes the current new data point, $\left(X^{*},{F}^{*},{S}^{*},\hat{D}^*\right)$, and checks if replacing the new data point with any data point currently in the history stack increases the minimum singular value. If the new data point does increase the minimum singular value, that is,
	\begin{equation}
	\lambda_{\min}\!\left(\!\sum_{i\neq j}{S}_{i}{S}_{i}^{T}\!+\!{S}_{j}{S}_{j}^{T}\!\right)\!<\frac{\!\lambda_{\min}\!\left(\!\sum_{i\neq j}{S}_{i}{S}_{i}^{T}\!+\!{S}^{*}{S}^{*T}\!\right)}{\left(1+\psi\right)},\label{eq:NLSsmax}
	\end{equation}
	where $\lambda_{\min}$ represents the minimum singular value of a matrix and $\psi$ is a positive constant, then
	the new data point replaces the data point currently in the $\mathcal{H}^{PE}$ that results in the largest increase in the minimum singular value, if not the new point is discarded.
	
	Using Lyapunov arguments, it can be shown (see Section \ref{sec: Ana PE}) that the parameter estimation error is directly related to the error $\mathcal{E}_i$ in \eqref{eq:NLSHistory Stack Compatibility}. Using the fact that newer values of $\hat{D}_i$ result in smaller $\mathcal{E}_i$ due to the exponential convergence of the disturbance estimates, a purging algorithm is developed to eliminate inaccurate data from $\mathcal{H}^{PE}$.
	
	The algorithm maintains two history stacks, a main history stack and a transient history stack, labeled $\mathcal{H}^{PE}$ and $\mathcal{G}^{PE}$, respectively. As soon as $\mathcal{G}^{PE}$ is full and sufficient time has elapsed since the last purge (see Section \ref{sec: Ana PE}), $\mathcal{H}^{PE}$ is emptied and $\mathcal{G}^{PE}$ is copied into $\mathcal{H}^{PE}$.
	
	\subsection{Analysis}\label{sec: Ana PE}
	A Lyapunov based analysis, summarized in the following theorem, is performed to show convergence for the parameter estimator developed in Section \ref{sec: Param Estim}.
	\begin{remark}
		To facilitate the following Lyapunov analysis, the dynamics for the parameter estimation error can be expressed as
		\begin{equation}
		\dot{\tilde{\theta}}=-\alpha_\theta\Gamma_\theta\mathscr{S}\tilde{\theta}-\alpha_\theta\Gamma_\theta\sum_{i=1}^{M}S_{i}\mathcal{E}_{i},\label{eq:thetaerrordyns}
		\end{equation}
		by using \eqref{eq:NLSHistory Stack Compatibility} and \eqref{eq:NLSTheta Dynamics}, along with the error being defined as $\tilde{\theta}=\theta-\hat{\theta}$.
	\end{remark}
The stability result is summarized in the following theorem.
	\begin{thm}\label{TM: PE}
		Provided the sequences of history stacks, $\mathcal{H}_1^{PE},\mathcal{H}_2^{PE},\ldots,$ are uniformly full rank\footnote{The authors definition of uniformly full rank history stacks $\mathcal{H}^{PE}_s$ requires a constant lower bound on $\underline{c}$ in \eqref{eq:NLRank Condition} $\forall s\in\mathbb{N}.$} and $\tilde{d}$ converges to zero exponentially, then for time intervals $[T_s,T_{s+1}) \ \forall s\in\mathbb{N}$, as $s \to \infty, \Vert\tilde{\theta}(T_s)\Vert \to 0$.
	\end{thm}	
	\begin{proof}
		Consider the candidate Lyapunov function
		\begin{equation}
		V_\theta(\tilde{\theta},t)=\frac{1}{2}\tilde{\theta}^T\Gamma_\theta^{-1}(t)\tilde{\theta}. \label{eq:V_estimation_theta}
		\end{equation}
		Using the bounds in \eqref{eq:NLSStaFGammaBound_theta}, the candidate Lyapunov function satisfies 
		\begin{equation}
		\frac{1}{\overline{\Gamma}_\theta}\left\Vert \tilde{\theta}\right\Vert ^{2}\leq V_\theta\left(\tilde{\theta},t\right)\leq\frac{1}{\underline{\Gamma}_\theta}\left\Vert \tilde{\theta}\right\Vert ^{2}.\label{eq:NLSCLNoXDotVBounds-1_theta}
		\end{equation}
		The time-derivative of \eqref{eq:V_estimation_theta} results in
		\begin{equation}
		\label{eqn:Dot_V_estimation_theta}
		\dot{V}_\theta(\tilde{\theta},t) = \tilde{\theta}^T \Gamma_\theta^{-1}(t) \dot{\tilde{\theta}} + \frac{1}{2}\tilde{\theta}^T \dot{\Gamma}_\theta^{-1}(t) {\tilde{\theta}}.
		\end{equation}
		Using (\ref{eq:NLSTheta Dynamics}) and (\ref{eq:NLGamma Dynamics}), along with the identity $\dot{\Gamma}^{-1}=-\Gamma^{-1}\dot{\Gamma}\Gamma^{-1}$, $\dot{V}_\theta$ can be expressed as
		\begin{equation*}
		\label{eq:Dot_V_estimation_final}
		\dot{V}_\theta(\tilde{\theta},t)= - \frac{1}{2}\alpha_{\theta} \tilde{\theta}^T  \mathscr{S}  \tilde{\theta} -\frac{1}{2} \beta_\theta\tilde{\theta}^T\Gamma_\theta^{-1}(t){\tilde{\theta}}-\alpha_\theta\tilde{\theta}^T\sum_{i=1}^{M}S_{i}\mathcal{E}_{i}^T.
		\end{equation*}
		Using the Cauchy-Schwartz inequality, and bounds in (\ref{eq:NLRank Condition}) and (\ref{eq:NLSStaFGammaBound_theta}), $\dot{V}_\theta$ can be bounded by
		\begin{equation}
		\label{eq:Dot_V_estimation_bound2_theta}
		\dot{V}_\theta(\tilde{\theta},t)\leq -\frac{1}{2}\left(\alpha_{\theta}\underline{c} + \frac{\beta_\theta}{\overline{\Gamma}_\theta}\right)\left\Vert\tilde{\theta}\right\Vert^2 +\alpha_\theta\left\Vert\tilde{\theta}\right\Vert\sum_{i=1}^{M}\left\Vert S_{i}\right\Vert\left\Vert\mathcal{E}_{i}\right\Vert.
		\end{equation}
	Since the states and controls are bounded, $\left\Vert S_{i}\right\Vert$ is bounded for all $i$. The upper bound is defined as
	\begin{equation} \label{eq: S Bound}
	\overline{S}:=M\max\{S_i\}, \forall i \in [1,\ldots,M].
	\end{equation}
	Using this upper bound and Young's Inequality, $\dot{V}_\theta$ becomes
	\begin{equation}
	\label{eq: Differential V}
	\dot{V}_\theta(\tilde{\theta},t)\leq -AV_\theta\left(\tilde{\theta},t\right) +B,
	\end{equation}
	where $A$ and $B$ are defined as
	\begin{equation}
	A:=\frac{\underline{\Gamma}_\theta}{4}\left(\alpha_{\theta}\underline{c} + \frac{\beta_\theta}{\overline{\Gamma}_\theta}\right),
	\end{equation}
	\begin{equation}
	B:=\frac{(\alpha_\theta\overline{S}\sum_{i=1}^{M}\left\Vert\mathcal{E}_i\right\Vert)^2}{\alpha_{\theta}\underline{c}+\nicefrac{\beta_{\theta}}{\overline{\Gamma}_\theta}}.
	\end{equation}
	Due to the purging of the data to remove erroneous estimates $\hat{d}$ from $\mathcal{H}^{PE}$, further analysis is needed to show parameter convergence. Let the purging instances be defined as $T_1,T_2,\ldots$ that maintain a minimum dwell time, $\mathcal{T}$, such that $T_{s+1}-T_s\geq\mathcal{T}>0, \ \forall s\in\mathbb{N}$.
	
	Solve equation $\eqref{eq: Differential V}$ over any time interval $[T_s,T_{s+1})$, yields
	\begin{equation}\label{eq: V Equation}
	\overline{V}_{s+1}\leq \overline{V}_se^{-A\left(t-T_s\right)}+\frac{B_{s+1}}{A},
\end{equation}
	where $\overline{V}_s\geq\left\Vert V_\theta\left(\tilde{\theta}\left(T_{s}\right),T_{s}\right)\right\Vert$ and $B_{s+1}$ denotes the value of $B$ over interval $[T_s,T_{s+1})$. Due to the exponentially decreasing error term $\left\Vert\mathcal{E}_i\right\Vert$, it can be seen that
	\begin{equation}\label{eq: Decreasing B}
	B_s> B_{s+1}, \forall s = 1,2,\ldots,
	\end{equation}
	and $\lim\limits_{s\to\infty}B_s=0.$ Furthermore, the dwell time condition results in the bound
	\begin{align*}
	\overline{V}_{s+1}\leq \overline{V}_se^{-A\mathcal{T}}+\frac{B_{s+1}}{A}, \forall s = 0,1,2,3,\ldots,
	\end{align*}
	If the bounds $B_s$ are selected so that 
	\begin{equation}\label{eq:ref}
	B_{s+1}>2B_se^{-A\mathcal{T}},\forall s = 0,1, \ldots,
	\end{equation} then
	\begin{equation}
	\overline{V}_{s+1}\leq \frac{2B_{s+1}}{A}, \forall s = 0,1,2,\ldots,
	\end{equation}
	where $B_0:=\frac{A\overline{V_0}}{2}$. As a result, $\lim\limits_{s\to\infty}\overline{V}_s=0$. It can further be concluded that $\left\Vert\tilde{\theta}\left(T_s\right)\right\Vert\to0$ as $s\to\infty$.

	\end{proof}
		\begin{remark}
		There is no loss of generality in assuming \eqref{eq:ref} since the bounds $B_s$ can be artificially inflated to meet \eqref{eq:ref}.
	\end{remark}
\section{Inverse Reinforcement Learning}\label{sec: IRL}

	The formulation of IRL in the following two sections closely follows the authors' previous work in \cite{SCC.Self.Harlan.ea2019a}. In addition, IRL is performed only on Agent 2, and the subscripts for the dynamics are omitted in the next sections.
	\subsection{Inverse Bellman Error}\label{sec:IBE}
	Under the premise that Agent 2 implements a feedback controller that would be optimal in a disturbance-free environment, the state and control trajectories, $x(\cdot)$ and $u(\cdot)$, satisfy the Hamilton-Jacobi-Bellman equation
	\begin{equation}
	H\left(x\left(t\right),\nabla_{x}\left(V^{*}\left(x\left(t\right)\right)\right)^{T},u\left(t\right)\right)=0,\forall t\in \mathbb{R}_{\geq 0},\label{eq:inverse HJB}
	\end{equation}
	where the unknown optimal value function is $ V^{*}:\mathbb{R}^{n}\to \mathbb{R} $ and $H:\mathbb{R}^n\times \mathbb{R}^n \times \mathbb{R}^m \to \mathbb{R}$ is the Hamiltonian, defined as $H(x,p,u):=p^Tf(x,u)+r(x,u)$. The goal of IRL is to accurately estimate the reward function, $ r $. To aid in the estimation of the reward function, let $ \hat{V}:\mathbb{R}^{n}\times \mathbb{R}^{P} \to \mathbb{R}$, $ \left(x,\hat{W}_{V}\right)\mapsto \hat{W}_{V}^{T}\sigma_{V}\left(x\right)+\epsilon_V(x) $ be a parameterized estimate of the optimal value function $V^{*}$, where $ \hat{W}_{V}\in \mathbb{R}^{P} $ are the estimates of the ideal value function weights $W_V^*$, $ \sigma_{V}:\mathbb{R}^{n}\to\mathbb{R}^{P} $ are known continuously differentiable features, and $\epsilon_V:\mathbb{R}^n\to\mathbb{R}$ is the resulting approximation error. Using $ \hat{\theta} $, $ \hat{W}_{V} $, $ \hat{W}_{Q} $, and $ \hat{W}_{R} $, which are the estimates of $ \theta $, $ W_{V}^{*} $, $ W_{Q}^{*} $, and $ W_{R}^{*}\coloneqq\left[r_{1},\cdots,r_{m}\right]^{T} $, respectively, in \eqref{eq:inverse HJB}, the inverse Bellman error $ \delta^{\prime}:\mathbb{R}^{n}\times \mathbb{R}^{m}\times\mathbb{R}^{L+P+m}\times \mathbb{R}^{p}\to\mathbb{R} $ is obtained as
	\begin{align}
	\delta^{\prime}\left({x},u,\hat{W},\hat{\theta}\right)=&\hat{W}_{V}^{T}\nabla_{x}\sigma_{V}\left({x}\right) \ \hat{Y}({x},u,\hat{\theta})+\hat{W}_{Q}^{T}\sigma_{Q}\left({x}\right)\nonumber\\
	&+\hat{W}_{R}^{T}\sigma_{u}\left(u\right),\label{eq: IBE}
	\end{align}where $ \sigma_{u}\left(u\right)\coloneqq\left[u_{1}^{2},\cdots,u_{m}^{2}\right]$ and $\hat{Y}({x},u,\hat{\theta})=\left[f^o({x},u)+\hat{g}({x},u,\hat{\theta})\right]$ where $\hat{g}({x},u,\hat{\theta}):=\hat{\theta}^T\sigma({x,u})$ from (\ref{eq:NLSPre Integral Form}). Rearranging, \eqref{eq: IBE} becomes \begin{equation}
	\delta^{\prime}\left({x},u,\hat{W}^{\prime},\hat{\theta}\right)=\left(\hat{W}^{\prime}\right)^{T}\sigma^{\prime}\left({x},u,\hat{\theta}\right),\label{eq:inverse BE}
	\end{equation}where $ \hat{W}^{\prime} \coloneqq \left[\hat{W}_{V};\hat{W}_{Q};\hat{W}_{R}\right] $ and $ \sigma^{\prime}\left({x},u,\hat{\theta}\right)\coloneqq\left[\nabla_{x}\sigma_{V}\left({x}\right)\hat{Y}({x},u,\hat{\theta});\sigma_{Q}\left({x}\right);\sigma_{u}\left(u\right)\right] $.
	
	\subsection{Inverse Reinforcement Learning Formulation} \label{sec: IRL Form}
	Using the formulation of the inverse Bellman error in Section \ref{sec:IBE}, and control signals, trajectories, and parameter estimates stored in a history stack, denoted as $\mathcal{H}^{IRL},$ the inverse Bellman error, evaluated at time instances $t_1,t_2,\ldots, t_N$ can be formulated into matrix form 
	\begin{equation}
	\Delta^{\prime} = \hat{\Sigma}^{\prime} \hat{W}^{\prime},\label{homogeneous}
	\end{equation}where $ \Delta^{\prime}\coloneqq\left[\delta^{\prime}_{t}\left(t_{1}\right);\cdots;\delta^{\prime}_{t}\left(t_{N}\right)\right] $, $ \delta^{\prime}_{t}\left(t\right)\coloneqq \delta^{\prime}\left({x}\left(t\right),u\left(t\right),\hat{W}^{\prime},\hat{\theta}\left(t\right)\right)$, and $ \hat{\Sigma}^{\prime}\coloneqq\left[\left(\hat{\sigma}_{t}^{\prime}\right)^{T}\left(t_{1}\right);\cdots;\left(\hat{\sigma}_{t}^{\prime}\right)^{T}\left(t_{N}\right)\right] $. The IRL problem can then be solved by finding the solution of the linear system in \eqref{homogeneous}. Since \eqref{homogeneous} is a homogeneous system of linear equations, it can only be solved up to a scaling factor. Since optimal state and control trajectories are invariant with respect to scaling of the cost function, the scaling ambiguity in \eqref{homogeneous} is to be expected. Since optimal control behaviours are scale-invariant, there is no loss of generality in resolving the scale ambiguity by assigning a fixed known value to one of the reward function weights.
	
	Taking the first element of $ \hat{W}_{R} $ to be known, the inverse BE in \eqref{eq:inverse BE} can then be expressed as \begin{equation}
	\delta^{\prime}\left({x},u,\hat{W},\hat{\theta}\right)=\hat{W}^{T}\sigma^{\prime\prime}\left({x},u,\hat{\theta}\right) + r_{1}\sigma_{u1}\left(u\right),
	\end{equation}where $ \hat{W} \coloneqq \left[\hat{W}_{V};\hat{W}_{Q};\hat{W}_{R}^{-}\right],$ the vector $ \hat{W}_{R}^{-} $ denotes $ \hat{W}_{R} $ with the first element removed, $\sigma_{ui}\left(u\right) $ denotes the $ i $\textsubscript{th} element of the vector $ \sigma_{u}\left(u\right) $, the vector $ \sigma_{u}^{-} $ denotes $ \sigma_{u} $ with the first element removed, and $ \sigma^{\prime\prime}\left({x},u,\hat{\theta}\right)\coloneqq\left[\nabla_{x}\sigma_{V}\left({x}\right)\hat{Y}({x},u,\hat{\theta});\sigma_{Q}\left({x}\right);\sigma_{u}^{-}\left(u\right)\right] $.
	
	The closed-form nonlinear optimal controller corresponding to the reward structure in (\ref{eq:reward}) provides the relationship\begin{multline}
	-2Ru\left(t\right)=\left(g^{\prime}(x(t))\right)^{T}(\nabla_{x}\sigma_{V}\left(x\left(t\right)\right))^TW_{V}^{*}\\+\left(g^{\prime}(x(t))\right)^{T}\nabla_{x}\epsilon\left(x\left(t\right)\right),
	\end{multline}
	which can be expressed as \begin{align*}
	-2r_{1}u_{1}\left(t\right) + \Delta_{u1}&=\sigma_{g1}\hat{W}_{V}\\
	\Delta_{u^{-}}&=\sigma_{g}^{-}\hat{W}_{V}+2\text{diag}\left(u_{2},\cdots,u_{m}\right)\hat{W}_{R}^{-},
	\end{align*}where $g^{\prime}(x):=\nabla_uf(x,u)$, $ \sigma_{g1} $ and $ \Delta_{u_{1}}$ denote the first rows and $ \sigma_{g}^{-}  $ and $ \Delta_{u^{-}} $ denote all but the first rows of $\sigma_{g}(x)\coloneqq \left(g^{\prime}(x)\right)^{T}(\nabla_{x}\sigma_{V}\left(x\right))^T $ and $ \Delta_u(x) \coloneqq \left(g^{\prime}(x)\right)^{T} \nabla_{x}\epsilon(x)$, respectively, and $ R^{-}\coloneqq\text{diag}\left(\left[r_{2},\cdots,r_{m}\right]\right) $. For simplification, let $ \sigma\coloneqq\left[\sigma^{\prime\prime},\, \begin{bmatrix}
	\sigma_{g}^{T}\\\Theta
	\end{bmatrix}\right]$, where \[ \Theta\coloneqq\left[0_{m\times 2n}, \,\,\begin{bmatrix}
	0_{1\times m-1}\\2\text{diag}\left(\left[u_{2},\cdots,u_{m}\right]\right)
	\end{bmatrix}\right]^{T}. \]
	
	Updating matrix form in (\ref{homogeneous}) by removing the known reward weight results in the linear system \begin{equation}
	-\Sigma_{u1}=\hat{\Sigma}\hat{W}-\Delta^{\prime}, \label{eq: Weight Linear System}
	\end{equation}where $ \hat{\Sigma}\coloneqq\left[\hat{\sigma}_{t}^{T}\left(t_{1}\right);\cdots;\hat{\sigma}_{t}^{T}\left(t_{N}\right)\right] $, and $ \Sigma_{u1}\coloneqq\left[\sigma_{u1}^{\prime}\left(u\left(t_{1}\right)\right);\cdots;\sigma_{u1}^{\prime}\left(u\left(t_{N}\right)\right)\right] $, where $ \hat{\sigma}_{t}\left(\tau\right)\coloneqq\sigma\left({x}\left(\tau\right),u\left(\tau\right),\hat{\theta}\left(\tau\right)\right) $, $ \sigma_{u1}^{\prime}(\tau)\coloneqq \left[r_1\sigma_{u1}(\tau);2r_{1}u_{1}(\tau);0_{\left(m-1\right)\times 1}\right]$. 
	
	At any time instant $ t $, provided the data stored in the history stack $\mathcal{H}^{IRL}$ satisfies \begin{equation}
	\textnormal{rank}\left(\hat{\Sigma}\right)=L+P+m-1,\label{eq:Rank Condition}
	\end{equation}
then the recursive update law 
\begin{equation}
\dot{\hat{W}}=\alpha\Gamma\hat{\Sigma}^T\left(-\hat{\Sigma}\hat{W}-\Sigma_{u1}\right),\label{eq:W Dynamics}
\end{equation}
is shown to result in UUB estimation of the weights $W^*$. In \eqref{eq:W Dynamics}, $\alpha\in\mathbb{R}_{>0}$ is a constant adaptation gain and $\Gamma:\mathbb{R}_{\geq0}\to\mathbb{R}^{\left(L+P+m-1\right)\times \left(L+P+m-1\right)}$ is the least-squares gain updated using the update law
\begin{equation}
\dot{\Gamma}=\beta\Gamma-\alpha\Gamma\hat{\Sigma}^T\hat{\Sigma}\Gamma,\label{eq:WGamma Dynamics}
\end{equation}
where $\beta\in\mathbb{R}_{>0}$ is the forgetting factor. Using arguments similar to \cite[Corollary 4.3.2]{SCC.Ioannou.Sun1996}, it can be shown that provided $\lambda_{\min}\left\{ \Gamma^{-1}\left(0\right)\right\} >0$, the least squares gain matrix satisfies
\begin{equation}
\underline{\Gamma}\text{I}_{L+P+m-1}\leq\Gamma\left(t\right)\leq\overline{\Gamma}\text{I}_{L+P+m-1},\label{eq:NLSStaFGammaBound}
\end{equation}
where $\underline{\Gamma}$ and $\overline{\Gamma}$ are positive
constants, and $\text{I}_{n}$ denotes an $n\times n$ identity matrix.
	\subsection{Analysis}\label{sec:Stability}
A Lyapunov based analysis is performed to show convergence for the IRL method in Section \ref{sec: IRL Form}.

	\begin{defn}
	A sequence of history stacks, $\mathcal{H}^{IRL}_1,\mathcal{H}^{IRL}_2,\ldots$, is called uniformly full rank if there exists a constant non-zero lower bound on all of the minimum singular values. More specifically,
	\begin{equation}
	0<\underline{\sigma}<\lambda_{\min}\left\{\hat{\Sigma}_s^T\hat{\Sigma}_s\right\},\forall s\in\mathbb{N},\label{def: Full Rank}
	\end{equation}
	where $\underline{\sigma}\in \mathbb{R}_{>0}.$ 
\end{defn}
\begin{remark}
	To facilitate the following Lyapunov analysis, the dynamics for the weight estimation error can be described by
	\begin{equation}
	\dot{\tilde{W}}=-\alpha\Gamma\hat{\Sigma}^T\left(\hat{\Sigma}\tilde{W}+\Delta_{\theta}\right),\label{eq:Werrordyns}
	\end{equation}
	using the fact that $\tilde{W}=W^*-\hat{W}$, along with \eqref{eq:W Dynamics} and the equation $-\Sigma_{u1}=\hat{\Sigma}W^*+\Delta_{\theta}$, where $\Delta_{\theta}$ denotes the errors resulting from poor $\hat{\theta}$ estimates incorporated in $\hat{\Sigma}$. 
\end{remark}

The stability result is summarized in the following theorem.

\begin{thm}
	Provided $\mathcal{H}^{PE}$ and $\mathcal{H}^{IRL}$ are uniformly full rank and $\tilde{d}$ converges to zero exponentially, then as $t \to \infty,$ $\Vert\tilde{W}\left(t\right)\Vert $ is uniformly ultimately bounded (UUB).
\end{thm}	
\begin{proof}
	Consider the positive definite candidate Lyapunov function
	\begin{equation}
	V(\tilde{W},t)=\frac{1}{2}\tilde{W}^T\Gamma^{-1}\left(t\right)\tilde{W}. \label{eq:V_estimation}
	\end{equation}
	Using the bounds in \eqref{eq:NLSStaFGammaBound}, the candidate Lyapunov function satisfies 
	\begin{equation}
	\underline{v}\left\Vert \tilde{W}\right\Vert ^{2}\leq V\left(\tilde{W},t\right)\leq\overline{v}\left\Vert \tilde{W}\right\Vert ^{2}.\label{eq:NLSCLNoXDotVBounds-1}
	\end{equation}
	where $\underline{v}:=\nicefrac{1}{2\overline{\Gamma}} $ and $\overline{v}:=\nicefrac{1}{2\underline{\Gamma}} $.
	
	The time-derivative of \eqref{eq:V_estimation} results in
	\begin{equation}
	\label{eqn:Dot_V_estimation}
	\dot{V}(\tilde{W},t) =\tilde{W}^T\Gamma^{-1}\left(t\right)\dot{\tilde{W}}+\frac{1}{2}\tilde{W}^T\dot{\Gamma}^{-1}\left(t\right)\tilde{W}.
	\end{equation}
	Using \eqref{eq:WGamma Dynamics} and \eqref{eq:Werrordyns}, along with the identity $\dot{\Gamma}^{-1}=-\Gamma^{-1}\dot{\Gamma}\Gamma^{-1}$, after simplifying the time-derivative can be expressed as
	\begin{multline*}
	\!\!\!\!\!\!\!\dot{V}(\tilde{W},t)\!=\!-\frac{1}{2}\alpha\tilde{W}^T\hat{\Sigma}^T\hat{\Sigma}\tilde{W}-\alpha\tilde{W}^T\hat{\Sigma}^T\!\Delta_{\theta}-\frac{1}{2}\beta\tilde{W}^T\Gamma^{-1}\!\!\left(t\right)\!\tilde{W}.
	\end{multline*}
	Substituting in $\hat{\Sigma}=\Sigma-\tilde{\Sigma}$
	\begin{multline*}
	\!\!\!\dot{V}(\tilde{W},t)= -\frac{1}{2}\alpha\tilde{W}^T\hat{\Sigma}^T\hat{\Sigma}\tilde{W}-\alpha\tilde{W}^T\left(\Sigma-\tilde{\Sigma}\right)^T\Delta_{\theta}\\-\frac{1}{2}\beta\tilde{W}^T\Gamma^{-1}\left(t\right)\tilde{W}.
	\end{multline*} 
	Since $\Delta_\theta=\Sigma W^*+\Delta_\epsilon-\hat{\Sigma}W^*$, substituting in and simplifying yields
	\begin{multline*}
	\!\!\!\dot{V}(\tilde{W},t)= -\frac{1}{2}\alpha\tilde{W}^T\hat{\Sigma}^T\hat{\Sigma}\tilde{W}-\frac{1}{2}\beta\tilde{W}^T\Gamma^{-1}\left(t\right)\tilde{W}\\-\alpha\tilde{W}^T\Sigma^T\tilde{\Sigma}W^*\!+\alpha\tilde{W}^T\tilde{\Sigma}^T\tilde{\Sigma}W^*\!\!-\alpha\tilde{W}^T\Sigma^T\Delta_\epsilon+\alpha\tilde{W}^T\tilde{\Sigma}^T\Delta_\epsilon.
	\end{multline*} 
	Using the Cauchy-Schwartz inequality, and bounds in (\ref{eq:NLSStaFGammaBound}) and \eqref{def: Full Rank}, $\dot{V}$ can be bounded by
	\begin{multline}
	\label{eq:Dot_V_estimation_bound2}
	\!\!\!\!\!\!\dot{V}(\tilde{W},t)\!\leq\! -\frac{1}{2}\!\left(\alpha\underline{\sigma}+\frac{1}{\overline{\Gamma}}\beta\right)\!\left\Vert\tilde{W}\right\Vert^2\!\!+\alpha\left\Vert\tilde{W}\right\Vert\left\Vert\Sigma\right\Vert\left\Vert\tilde{\Sigma}\right\Vert\left\Vert W^*\right\Vert\\+\alpha\left\Vert\tilde{W}\right\Vert\left\Vert\tilde{\Sigma}\right\Vert^2\left\Vert W^*\right\Vert+\alpha\left\Vert\tilde{W}\right\Vert\left\Vert\Sigma\right\Vert\left\Vert\Delta_{\epsilon}\right\Vert\\+\alpha\left\Vert\tilde{W}\right\Vert\left\Vert\tilde{\Sigma}\right\Vert\left\Vert\Delta_{\epsilon}\right\Vert.
	\end{multline}
	Based on the linearly parameterized reward weights, the norm of the resulting error term $\Delta_{\epsilon}$ can be expressed as
	\[
	\!\!\!\!\overline{\Delta}_{\epsilon}:=\left(N\max_{\substack{ x \in x\left(\cdot\right) \\ u \in u\left(\cdot\right)}}\left\{\epsilon_V^2\left(x\right)+\epsilon_Q^2\left(x\right)+\epsilon_u^2\left(u\right)\right\}\right)^{\nicefrac{1}{2}}.
	\]
	Using this upper bound, $\dot{V}$ becomes
	\begin{multline}
	\!\!\dot{V}(\tilde{W},t)\leq -\frac{1}{2}\left(\alpha\underline{\sigma}+\frac{1}{\overline{\Gamma}}\beta\right)\left\Vert\tilde{W}\right\Vert^2\\+\alpha\overline{\Delta}_\epsilon\left\Vert\tilde{W}\right\Vert\left\Vert{\Sigma}\right\Vert+\alpha\overline{\Delta}_\epsilon\left\Vert\tilde{W}\right\Vert\left\Vert\tilde{\Sigma}\right\Vert\\+\alpha\left\Vert\tilde{W}\right\Vert\left\Vert\Sigma\right\Vert\left\Vert\tilde{\Sigma}\right\Vert\left\Vert W^*\right\Vert+\alpha\left\Vert\tilde{W}\right\Vert\left\Vert\tilde{\Sigma}\right\Vert^2\left\Vert W^*\right\Vert.
	\end{multline}
	The term $\left\Vert\tilde{\Sigma}\right\Vert$ can be expressed in terms of $\tilde{\theta}$ as
	\begin{equation}
	\left\Vert\tilde{\Sigma}\right\Vert\leq\left\Vert\tilde{\theta}\right\Vert \overline{\Sigma}, \label{eq: Sigma Tilde Bound}
	\end{equation}
	where
	\begin{equation}
	\overline{\Sigma}:=N\max_{\substack{ x \in x\left(\cdot\right) \\ u \in u\left(\cdot\right)}}\{\left\Vert\nabla_x\sigma_{V}(x)\right\Vert\left\Vert\sigma(x,u)\right\Vert\}.
	\end{equation}
	The term $\left\Vert\Sigma\right\Vert$, which contains true values of the unknown parameters, is bounded above since it is a function of only true parameters, $\theta$, and bounded states and controls, $x$ and $u$. Let the upper bound on $\left\Vert\Sigma\right\Vert$ be denoted as
	\begin{equation}\label{eq: Sigma Bound}
	 \left\Vert\Sigma\right\Vert\leq\overline{\Sigma}_\theta.
	\end{equation}
	
	Using \eqref{eq: Sigma Tilde Bound} and \eqref{eq: Sigma Bound}, $\dot{V}$ becomes
	\begin{multline}
	\!\!\dot{V}(\tilde{W},t)\leq -\frac{1}{2}\left(\alpha\underline{\sigma}+\frac{1}{\overline{\Gamma}}\beta\right)\left\Vert\tilde{W}\right\Vert^2+\alpha\overline{\Delta}_\epsilon\overline{\Sigma}_\theta\left\Vert\tilde{W}\right\Vert\\+\alpha\overline{\Delta}_\epsilon\overline{\Sigma}\left\Vert\tilde{W}\right\Vert\left\Vert\tilde{\theta}\right\Vert+\alpha\overline{\Sigma}_\theta\overline{\Sigma}\left\Vert W^*\right\Vert\left\Vert\tilde{W}\right\Vert\left\Vert\tilde{\theta}\right\Vert\\+\alpha\overline{\Sigma}^2\left\Vert W^*\right\Vert\left\Vert\tilde{W}\right\Vert\left\Vert\tilde{\theta}\right\Vert^2.
	\end{multline}
	Using Young's Inequality $\dot{V}$ then becomes
	\begin{multline}
	\label{eq:Dot_V_estimation_Final bound}
	\!\!\dot{V}(\tilde{W},t)\leq  -\frac{1}{8}\left(\alpha\underline{\sigma}+\frac{1}{\overline{\Gamma}}\beta\right)\left\Vert\tilde{W}\right\Vert^2\\\!\!\!+\frac{2\left(\alpha\overline{\Delta}_\epsilon\overline{\Sigma}_\theta\overline{\tilde{\theta}}+\alpha\overline{\Sigma}_\theta\overline{\Sigma}\left\Vert W^*\right\Vert\overline{\tilde{\theta}}+\alpha\overline{\Sigma}^2\left\Vert W^*\right\Vert\overline{\tilde{\theta}}^2\right)^2}{\alpha\underline{\sigma}+\nicefrac{\beta}{\overline{\Gamma}}}\\+\frac{2(\alpha\overline{\Delta}_\epsilon\overline{\Sigma}\left\Vert\theta\right\Vert)^2}{\alpha\underline{\sigma}+\nicefrac{\beta}{\overline{\Gamma}}},
	\end{multline}
	where $\overline{\tilde{\theta}}$ denotes bounded $\tilde{\theta}$ values stored in the history stack, $\mathcal{H}^{IRL}$.
	Using the bound in \eqref{eq:NLSCLNoXDotVBounds-1}, the differential inequality for $\dot{V}$ can be expressed as
	\begin{equation}
	\label{eq: Differential V of Z}
	\!\!\dot{V}(\tilde{W},t)\leq -AV\left(\tilde{W},t\right)+B+C,
	\end{equation}
	where
	\begin{equation*}
	A:=\frac{1}{8\overline{v}}\left(\alpha\underline{\sigma}+\frac{1}{\overline{\Gamma}}\beta\right),
	\end{equation*}
	\begin{equation}
	B:=\frac{2\left(\alpha\overline{\Delta}_\epsilon\overline{\Sigma}_\theta\overline{\tilde{\theta}}+\alpha\overline{\Sigma}_\theta\overline{\Sigma}\left\Vert W^*\right\Vert\overline{\tilde{\theta}}+\alpha\overline{\Sigma}^2\left\Vert W^*\right\Vert\overline{\tilde{\theta}}^2\right)^2}{\alpha\underline{\sigma}+\nicefrac{\beta}{\overline{\Gamma}}},
	\end{equation} and
	\begin{equation}
	C:=\frac{2(\alpha\overline{\Delta}_\epsilon\overline{\Sigma}\left\Vert\theta\right\Vert)^2}{\alpha\underline{\sigma}+\nicefrac{\beta}{\overline{\Gamma}}}.
	\end{equation}
	Due to the fact that $\hat{\Sigma}$ and $\Delta'$ depend on the quality of the parameter estimates, a purging technique was incorporated in an attempt to remove poor estimates $\hat{\theta}$ from $\mathcal{H}^{IRL}.$ During the transient phase of the parameter estimator, the estimates $\hat{\theta}$ are less accurate and the resulting values of $\hat{W}$ will be poor. Purging facilitates usage of better estimates as they become available.
	
		Due to purging of $\mathcal{H}^{IRL}$, the estimator is analyzed over discrete time instances. Define the purging instances as $T_1,T_2,\ldots$, and maintain a minimum dwell time, $\mathcal{T}$, such that $T_{s+1}-T_s\geq\mathcal{T}>0, \ \forall s\in\mathbb{N}$.
	
	Solving equation $\eqref{eq: Differential V of Z}$ over any time interval $[T_s,T_{s+1})$, yields
	\begin{equation}
	\overline{V}_{s+1}\leq \overline{V}_se^{-A\left(t-T_s\right)}+\frac{B_s}{A}+\frac{C}{A},
	\end{equation}
	where $\overline{V}_s\geq\left\Vert V\left(\tilde{W}\left(T_{s}\right),T_{s}\right)\right\Vert$ and $B_{s+1}$ denotes the value of $B$ over interval $[T_s,T_{s+1})$. 
	A similar argument as the proof of Theorem \ref{TM: PE} can be used to conclude that
	\begin{equation}
	\lim\limits_{s \to \infty}\sup\overline{V}_s\leq\frac{32\overline{v}(\alpha\overline{\Delta}_\epsilon\overline{\Sigma}\left\Vert\theta\right\Vert)^2}{\left(\alpha\underline{\sigma}+\nicefrac{\beta}{\overline{\Gamma}}\right)^2},
	\end{equation}and as a result  $\lim\limits_{s\to\infty}\sup\left\Vert\tilde{W}\left(T_s\right)\right\Vert\leq\sqrt{2\frac{\overline{v}}{\underline{v}}}\frac{4(\alpha\overline{\Delta}_\epsilon\overline{\Sigma}\left\Vert\theta\right\Vert)}{\left(\alpha\underline{\sigma}+\nicefrac{\beta}{\overline{\Gamma}}\right)}$.
\end{proof}
		
	\section{Simulation}\label{sec:Sim}
	To demonstrate the performance of the developed method, a nonlinear optimal control problem was constructed using \cite{SCC.Kalman1964} in order to have a known value function for comparison. 
	
	Agent 1 has the following nonlinear dynamics
	\begin{equation}
	\dot{x}_{1_1}= x_{1_2}\nonumber, \ \ \ \ \ \ \ \ 
	\dot{x}_{1_2}= x_{1_1}x_{1_2}+3x_{1_2}^2+5u_1+d.
	\label{eq: Linear Simulation}
	\end{equation}
	Agent 2 under observation has the following nonlinear dynamics
	\begin{align}
	\dot{x}_{2_1}&= x_{2_2},\nonumber\\
	\ \ \dot{x}_{2_2}&= \theta_1x_{2_1} \Big(\frac{\pi}{2}+\tan^{-1}(5x_{2_1})\Big)+\frac{\theta_2x_{2_1}^2}{1+25x_{2_1}^2} \nonumber \\ & \ \ \ \ \ \ \ \  \ \ \ \ \ \ \ \ \ \ \ \ \  \ \ \ \ \ \  \ \ \ \  \ \ \  +\theta_3x_{2_2}+3u_2+d,
	\label{eq: NonLinear Simulation}
	\end{align}
	where $x_{A_B}$ denotes state B for Agent A. The parameters $\theta_1,\theta_2,$ and $\theta_3$ are unknown constants to be estimated and $d$ is the unknown disturbance. The exact values of these parameters are $\theta_1=-1,\theta_2=-\frac{5}{2},$ and $\theta_3=4$. The disturbance, $d$, acting on the agents is generated from the linear system in Section \ref{sec:Dist_Estimation}, where $A=[0,1;-1,0]$ and $C=[0,0;1,0]$, and the chosen gain matrix was $K=[1,0.5;0,5].$
%
	
	The performance index that the agent is trying to minimize is
	\begin{equation*}
	J(x_0,u_2(\cdot)) = \int_{0}^{\infty}(x_{2_2}^2+u_2^2)dt,
	\end{equation*}
	resulting in the reward function weights to be estimated as $Q=\mathrm{diag}(q_1,q_2)=\mathrm{diag}(0,1)$ and $R=1$. The observed state and control trajectories, and the disturbance estimates are used in the estimation of unknown parameters in the dynamics, along with the optimal value function parameters and the reward function weights. The optimal controller is $u_2^*=-3x_{2_2}$, while the optimal value function is $V^*=x_{2_1}^2(v_1+v_2\tan^{-1}(5x_{2_1}))+v_3x_{2_2}^2$,
%
	resulting in the ideal function parameters $v_1=\frac{\pi}{2},$ $v_2=1$, and $v_3=1$.
	\begin{figure}[htbp!]
		\includegraphics[width=1\columnwidth]{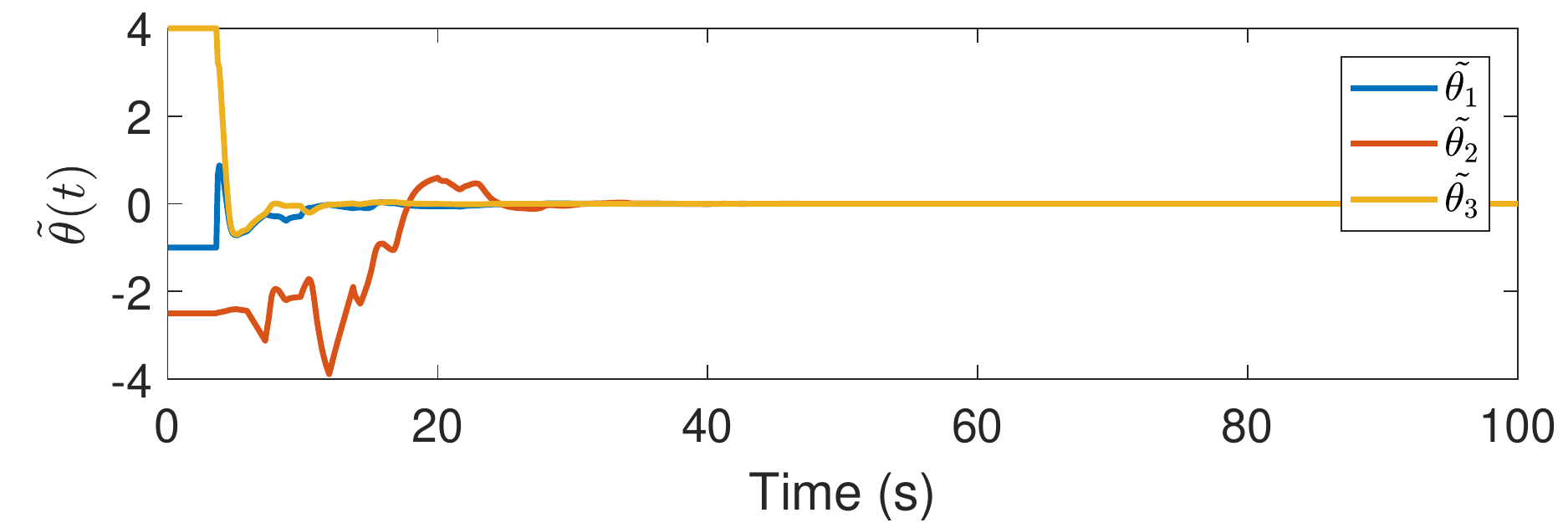}
		\caption{Estimation error for the unknown parameters in Agent 2's dynamics.}
		\label{fig: Parameter Estimation}
	\end{figure}
	\begin{figure}[htbp!]
		\includegraphics[width=1\columnwidth]{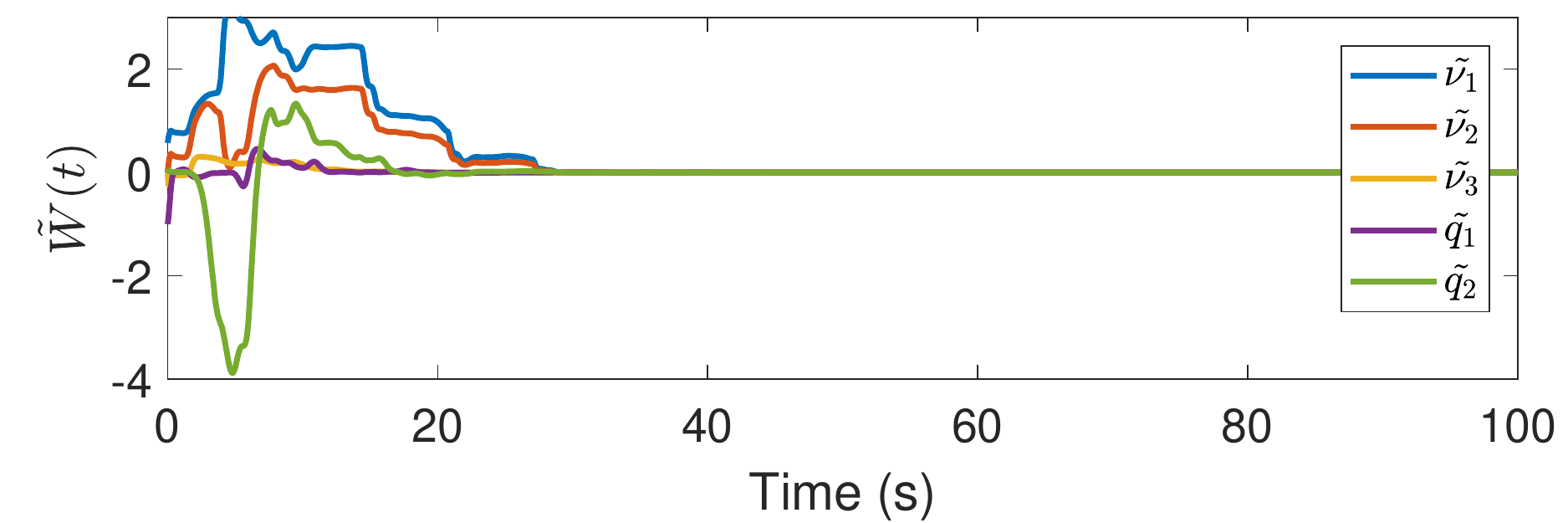}
		\caption{Estimation error for the unknown parameters in the reward function for Agent 2.}
		\label{fig: Cost Function Estimation}
	\end{figure}
		\begin{figure}[htbp!]
			\includegraphics[width=1\columnwidth]{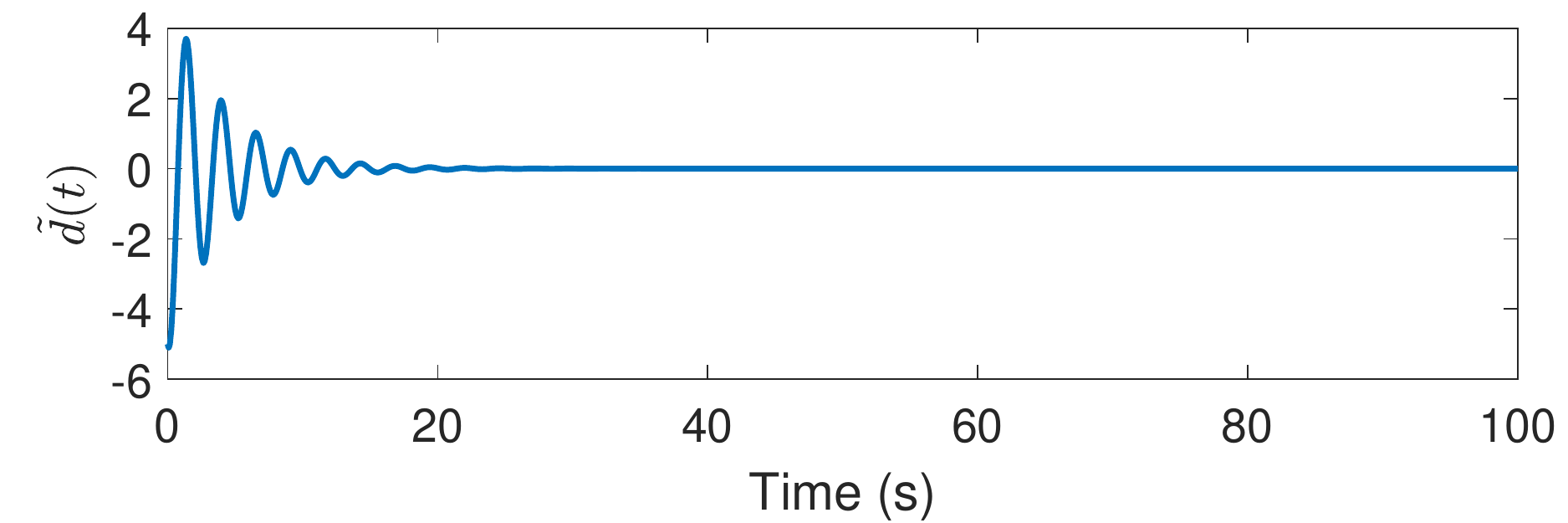}
			\caption{Estimation error for the unknown disturbance acting on the two agents.}
			\label{fig: Dist Estimation}
		\end{figure}
	Figs. \ref{fig: Parameter Estimation} and \ref{fig: Cost Function Estimation} show the performance of the proposed method. Fig. \ref{fig: Parameter Estimation} shows convergence of the unknown part of Agent 2's dynamics, and Fig. \ref{fig: Cost Function Estimation} shows convergence of the unknown reward function. Fig. \ref{fig: Dist Estimation} shows the convergence of the disturbance estimates. The parameters used for the simulation are: $T = 1.2 s$, $N = 100$, $M = 150$, $\beta=\beta_{\theta} = 0.5$, $\alpha=\alpha_{\theta} = \nicefrac{1}{N}$, and a time step of $0.0005s$.
	\section{Conclusion} \label{sec:Conclusion}
	A novel IRL framework is developed in this paper for reward function estimation in the presence of modeling errors and additive disturbances. To compensate for disturbance-induced sub-optimality of observed trajectories, a model-based approach is developed that relies on a disturbance estimator. 
	
	Future work will focus on the development of output feedback IRL methods that utilize both state and parameter estimation methods, and extensions of the developed method for disturbances that affect the agents through a control effectiveness matrix. The authors will additionally explore the use of implicit disturbance estimation techniques that would result in bounded disturbance estimation errors.
	\bibliographystyle{ieeetran}
	\bibliography{scc,sccmaster}
\end{document}